\documentclass[final,runningheads,envcountsame]{llncs}

\usepackage{amsmath}
\usepackage{amssymb}
\usepackage{colonequals}
\usepackage{adjustbox}
\usepackage[T1]{fontenc}
\usepackage{lmodern}
\usepackage{cite}
\usepackage{hyperref}
\usepackage{microtype}
\usepackage[capitalise]{cleveref}

\usepackage{tikz-cd}
\usepackage{proof}

\usepackage[only,leftrightarroweq]{stmaryrd}
\newcommand{\bisim}{\leftrightarroweq}
\newcommand{\bbisim}{\leftrightarroweq_{b}}
\newcommand{\apart}{\#}
\newcommand{\bapart}{\#_{b}}

\usepackage[colorinlistoftodos,prependcaption,textsize=tiny]{todonotes}

\usepackage{xargs}

\begin{document}

\title{Relating Apartness and \\Branching Bisimulation Games\thanks{%
This research is partially supported by the Royal Society International Exchange grant (IES\textbackslash R3\textbackslash 223092). The third author was also partially supported by EPSRC NIA grant EP/X019373/1.}}
\author{Jurriaan Rot\inst{1}, Sebastian Junges\inst{1} and Harsh Beohar\inst{2}}

\institute{
Institute for Computing and Information Sciences (iCIS), Radboud University, NL\and
Department of Computer Science, University of Sheffield, UK
}

\maketitle

\begin{abstract}
Geuvers and Jacobs (LMCS 2021) formulated the notion of apartness relation on state-based systems modelled as coalgebras. In this context apartness is formally dual to bisimilarity, and gives an explicit proof system for showing that certain states are not bisimilar. In the current paper, we relate apartness to another classical element of the theory of behavioural equivalences: that of turn-based two-player games. Studying both strong and branching bisimilarity, we show that winning configurations for the Spoiler player correspond to apartness proofs, for transition systems that are image-finite (in the case of strong bisimilarity) and finite (in the case of branching bisimilarity).
\end{abstract}

\section{Introduction}
Bisimilarity is one of the fundamental notions of equivalence~\cite{GlabbeekSpectrumII}, encoding when two states of a labelled transition system (LTS) have the same behaviour.
Bisimilarity is well studied in the literature from both logical and game-theoretic viewpoints.
For instance, the classical Hennessy-Milner characterisation theorem~\cite{hm:hm-logic} states that two states of an image-finite LTS are bisimilar if and only if they satisfy the same set of modal formulas. Similarly, the well-known result by Stirling~\cite{Stirling92:TReport:bisimGames}  states that two states of an LTS are bisimilar if and only if Duplicator has a winning strategy from this pair of states in the Spoiler/Duplicator bisimulation game. These two viewpoints have almost become a standard in the sense that it is expected that similar characterisation results hold, whenever a new notion of behavioural equivalence is proposed.

Orthogonally to these logical and game-theoretic viewpoints, in the recent work of Geuvers and Jacobs~\cite{GeuversJacobs21-apartness}, a dual approach to bisimilarity is postulated in terms of \emph{apartness} in transition systems.
Instead of describing when two states are behaviourally equivalent, as in bisimilarity,
the motive of an apartness relation is in showing \emph{differences} in behaviour. More formally, where bisimilarity is a coinductive characterisation of behavioural equivalence, apartness inductively provide a proof system for constructing witnesses of such differences. Geuvers and Jacobs propose a general coalgebraic formulation of apartness, and show how this yields concrete proof systems for deterministic automata, labelled transition systems and streams. They also develop versions of apartness for weak and branching bisimilarity.

This research strand allows us to study connections between modal logic, games and bisimilarity through the lens of apartness. In particular, the Hennessy-Milner theorem says that two states are apart if and only if there is a distinguishing formula, i.e., a formula that holds in one state but not the other. For games, a natural formulation is that two states are apart if and only if Spoiler has a winning strategy. Both results hold by simply observing that bisimilarity is the complement of apartness~\cite{GeuversJacobs21-apartness}. However, such an approach is rather implicit: it does not really show how to move between apartness proofs, distinguishing formulas and winning strategies for Spoiler.

The relation between apartness proofs and distinguishing formulas is studied in~\cite{Geuvers2022}, and revisited in an abstract coalgebraic setting in~\cite{TurkenburgBKR23}.
In the current paper, we focus on the relation between apartness and bisimulation games. One of the main messages of this paper is the following dichotomy: bisimulations correspond to the winning strategies for Duplicator, while apartness relations correspond to winning strategies for Spoiler.
We explicitly relate winning configurations for Spoiler to apartness proofs.
We first develop the correspondence between apartness relations and Spoiler strategies for strong bisimilarity, and then move on to branching bisimilarity~\cite{BBisim:96}, following the game characterisation in~\cite{YinFHHT14}. Our proofs rely on the assumption that Duplicator has only finitely many possible moves; this is true under the assumption that the LTS is image-finite, in the case of strong bisimilarity. For branching bisimilarity, since Duplicator can answer with a sequence of $\tau$ moves, we make a stronger assumption for our proof strategy to work: that the LTS is finite.

\section{Strong bisimilarity}\label{sec:strongbisim}

The objective of this section is to show that the winning strategy of Spoiler in a strong bisimulation game corresponds to an apartness relation. To this end, we first recall preliminaries on bisimulations~\cite{Park81-bisim}, games \cite{Stirling92:TReport:bisimGames,Stirling99-BisimGames}, and (proper) apartness relations~\cite{GeuversJacobs21-apartness,Geuvers2022}.

Throughout this section, we fix an LTS, which we take to be a tuple $(X, A, \rightarrow)$ consisting of a set $X$ of states, a set $A$ of actions, and a transition relation ${\rightarrow} \subseteq X \times A \times X$. We assume our LTS to be \emph{image-finite}: for every state $x \in X$ and label $a \in A$, the
set $\{x' \mid x \xrightarrow{a} x'\}$ is finite.

\begin{definition}
  A symmetric relation $R\subseteq X \times X$ is a \emph{(strong) bisimulation} if for all $(x,y) \in R$:
  \begin{itemize}
	  \item
  if $x \xrightarrow{a} x'$ then $\exists y'.\, y \xrightarrow{a} y' \wedge x' \mathrel R y'$.
  \end{itemize}
  Two states $x,y\in X$ are \emph{bisimilar}, denoted $x \bisim y$, iff there exists a bisimulation $R$ such that $x \mathrel R y$; i.e., $\bisim\ = \bigcup\{R \mid R \text{ is a bisimulation}\}$.
\end{definition}
The relation $\bisim$ is itself a bisimulation, and it is trivially the largest one.

\begin{definition}\label{def:apart}
	A symmetric relation $R \subseteq X \times X$ is called an \emph{apartness relation} if it satisfies:
  \[
  \frac{x \xrightarrow a x' \quad \forall y'.\, y \xrightarrow a y' \text{ implies }  x' \mathrel R y'}{x \mathrel R y}
  \]
  Two states $x,y\in X$ are \emph{apart}, denoted $x \apart y$, iff $(x,y)$ are related in every apartness relation $R$, i.e., $\apart = \bigcap \{R \mid R \text{ is an apartness relation}\}$. We refer to the relation $\apart$ simply as \emph{apartness}.
\end{definition}
Apartness $\apart$ is an apartness relation itself.
Since apartness is characterised as the \emph{least} relation satisfying the above rule, it provides a proof technique: showing that two states are apart amounts to giving a proof using the rule. Moreover, it suffices to consider \emph{finite} proofs, by the assumption that the LTS is image-finite.
\begin{example}\label{ex:apart}
	Consider the following LTS.
	$$
  
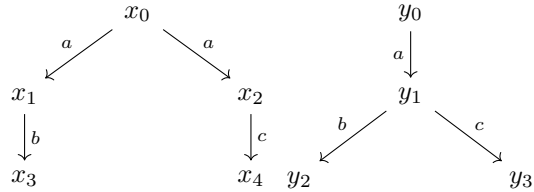
\begin{tikzcd}
    & \arrow[ld,"a"'] x_0 \arrow[rd,"a"] & \\
    x_1\arrow[d,"b"]  & & x_2\arrow[d,"c"]\\
    x_3 & & x_4
  \end{tikzcd}
  \begin{tikzcd}
    & \arrow[d,"a"'] y_0 & \\
    & \arrow[ld,"b"'] y_1 \arrow[rd,"c"]  & \\
    y_2 & & y_3
  \end{tikzcd}
  $$
  The states $x_0$ and $y_0$ are apart, which we can show with following proof tree, where we also explicitly use that the apartness is a symmetric relation:
  \[
  \infer{x_0 \mathrel \apart y_0}{x_0 \xrightarrow {a} x_1 &
  \infer{ \forall y'. y_0 \xrightarrow{a} y' \text{ implies }
  x_1 \mathrel \apart y'}
    {\infer{x_1 \mathrel \apart y_1}{\infer{y_1 \mathrel \apart x_1}
        {
        y_1 \xrightarrow{c} y_3 \quad \neg(x_1 \xrightarrow{c} )
        }
    }
  }}
  \]
  In the above proof, $\neg(x_1 \xrightarrow{c})$ means there is no transition of the form $x_1 \xrightarrow{c} x'$, that is, no outgoing $c$-transition from $x_1$. This means the universally quantified statement in the rule for apartness vacuously holds.
 \end{example}

\begin{theorem}[Geuvers and Jacobs~\cite{GeuversJacobs21-apartness}]\label{thm:complement-strong-bis}
Apartness is dual to bisimilarity, i.e., $\apart = (X \times X) \setminus \! \bisim$.
\end{theorem}

We next recall the strong bisimulation game of Stirling~\cite{Stirling92:TReport:bisimGames}, using notation from~\cite{Frutos-EscrigKW17}. The intuition is Duplicator wants to show that two states are bisimilar, while Spoiler wants to show the difference in their behaviour (i.e., they are apart; see \cref{def:apart}).
The game is a \emph{turn-based two-player game on a graph} $(C = C_D \uplus C_S,\rightarrow_D,\rightarrow_S)$ where the players are called Duplicator (or just $D$) and Spoiler (or just $S$). We refer to the nodes $C$ as \emph{configurations}, to Duplicator moves by ${\rightarrow_D} \subseteq C_D \times C_S$, and Spoiler moves by ${\rightarrow_S} \subseteq C_S \times C_D$.

\begin{definition}[Strong bisimulation game]\label{def:bisim-game}
   The set of configurations is given by $C = C_S \uplus C_D$, where $C_S = X \times X$ are Spoiler configurations, ranged over by tuples denoted by $[x,y]$, and $C_D = X \times A \times X \times X$ are Duplicator configurations, ranged over by tuples of the form $\langle x,a,x',y \rangle$. The moves are:
   \begin{itemize}
     \item Spoiler can move from a configuration $[x,y]$ as follows:
	 \begin{enumerate}
	 	\item to $\langle x,a,x',y \rangle$ if there is a transition of the form  $x \xrightarrow a x'$;
		\item to $\langle y,a,y',x \rangle$ if there is a transition of the form  $y \xrightarrow a y'$;
	\end{enumerate}
     \item Duplicator can move from a configuration $\langle x,a,x',y \rangle$ to $[x',y']$ if there exists a transition of the form $y \xrightarrow{a} y'$.
   \end{itemize}
   \end{definition}
   Spoiler wins a play if and only if Duplicator cannot move. To formalise this, we need a few definitions regarding two-player games.
  \begin{definition}We define plays and winning configurations as follows.
	\begin{itemize}
		\item
   A \emph{play} from a Spoiler configuration $[x,y]$ is a finite or infinite sequence $\sigma \in C^* \cup C^\omega$ of configurations such that $\sigma_0 = [x,y]$ and for all $i>0$, $\sigma_{i+1}$ is a move from $\sigma_i$.
  \item
  A play is \emph{maximal} if it is either infinite, or there is no move possible from the last configuration. A finite maximal play is winning for Spoiler if this last configuration is in $C_D$ (that is, Duplicator is stuck); all other maximal plays are won by Duplicator.
  \item A (positional) \emph{strategy} for player $P \in \{D,S\}$ is a map $\pi_P$ from configurations $C_P$ to moves for player $P$ in that configuration. A play $\sigma$ is \emph{consistent} with a strategy $\pi_P$ for player $P$ if for all $i$ such that $\sigma_i \in C_P$, we have that $\pi_P(\sigma_i) = \sigma_{i+1}$.
  \item A strategy $\pi_P$ for player $P$ is called \emph{winning from a configuration} $[x,y] \in C_S$ if every maximal play that is consistent with $\pi_P$ is winning for $P$. A Spoiler configuration $[x,y] \in C_S$ is winning for player $P$ if there exists a winning strategy from $[x,y]$ for that player.
  The set of winning configurations of player $P$ is called the \emph{winning region} of $P$, and is denoted by $\mathcal{W}_P$.
  \end{itemize}
\end{definition}


Bisimulation games are so-called reachability games for Spoiler, i.e., Spoiler wants to reach a particular set of configurations. Consequentially, they are  so-called safety games for Duplicator, who wants to ensure that we do not reach those configurations.
To make this explicit, we provide fixed-point characterisations of the winning regions of both players. Consider the usual definition of box ($\Box_D,\Box_S$) and diamond ($\lozenge_D,\lozenge_S$) modalities:
\begin{align*}
  \lozenge_D W =&\ \{c\in C_D \mid \exists c'.\ c \rightarrow_D c' \land c'\in W \} \\
  \Box_D W =&\ \{c\in C_D \mid \forall c'.\ c\rightarrow_D c' \text{ implies } c'\in W\},
\end{align*}
where $W\subseteq C$. The modalities $\lozenge_S$, $\Box_S$ are defined analogously.

\begin{proposition}\label{prop:WinRegions}
  Let $(C,\rightarrow_D,\rightarrow_S)$ be an alternating two-player game.
  \begin{itemize}
    \item The winning region $\mathcal{W}_D$ of Duplicator is the largest set $W$ such that
        \[ W \subseteq \Box_S \lozenge_D W\,.\]
    \item The winning region $\mathcal{W}_S$ of Spoiler is the least set $W$ such that
        \[ \lozenge_S \Box_D W \subseteq W\,.\]
  \end{itemize}	
\end{proposition}
\begin{proof}
		The fixed-point characterisation for winning region of Duplicator is studied in, e.g.,~\cite{KomoridaKHKH19-codensitygames}.
	  We consider the Spoiler case.
	  We first show that $\mathcal{W}_S$ indeed satisfies $\lozenge_S \Box_D \mathcal{W}_S \subseteq \mathcal{W}_S$.
	  To this end we introduce some notation: for a strategy $\pi$ and configuration $c$, we denote by $\pi[c]$ the
	  set of \emph{reachable} configurations, i.e., the set of all $c' \in C_S$ such that $c'$ occurs in a play from $c$ consistent with $\pi$. Note that
	  if $\pi$ is a winning strategy from $c$, then it is winning from all $c' \in \pi[c]$.
	
	  Now suppose $c_0 \in \lozenge_S \Box_D \mathcal{W}_S$, so there exists $c_0 \rightarrow_S c'$ such that for all configurations $c''$ with $c' \rightarrow_D c''$, $c''$ is winning for Spoiler, i.e., $c'' \in \mathcal{W}_S$. We refer to these resulting Spoiler winning configurations as $c_1, \ldots, c_k$ (note that this is a finite set, since the LTS is image-finite; although this is not strictly needed for this side of the argument), and associate winning strategies $\pi_1, \ldots, \pi_k$ for each of these configurations. 
	  We construct a strategy $\pi$ as follows:\footnote{The challenge in combining the strategies $\pi_1, \ldots, \pi_k$ is that there may be overlap between the reachable sets $\pi_1[c_1], \ldots, \pi_k[c_k]$. This issue arises specifically because we only consider \emph{positional} strategies; if the first move (i.e., a choice of configuration $c_i$) is recorded, one can stick to the corresponding strategy $\pi_i$.}
	  $$
	  \pi(c) =
	  \begin{cases}
	  	\pi_{i}(c) & \text{if } c \in \pi_i[c_i] \text{ and for all }j< i. \, c \not \in \pi_j[c_j]  \\
		c' & \text{if } c=c_0 \text{ and }\forall j. \, c \not \in \pi_j[c_j] \\
		\text{any }c_* \in C_D & \text{otherwise}
	  \end{cases}
	  $$
	  This is a winning strategy from $c_0$: after one move from Spoiler, every move from Duplicator leads us to a configuration for which the strategy is winning.
	
	  Second, suppose $W$ is an arbitrary set such that $\lozenge_S \Box_D W \subseteq W$.
	  We have to show that $\mathcal{W}_S \subseteq W$.
	  This is where we crucially rely on the assumption that our LTS is image-finite.
	  Let $c \in \mathcal{W}_S$ and let $\pi$ be a winning strategy from $c$ (see \cref{def:bisim-game}). Consider the tree of all possible plays from $c$ consistent with $\pi$; the edges in this tree are given by Duplicator moves. Since the LTS is image-finite, Duplicator has finitely many choices at each point and therefore this tree is finitely branching. The crux is that it is also finite depth; for suppose it is not, then by K\"onig's lemma there would be an infinite play, which is winning for Duplicator and contradicts that $\pi$ is winning for Spoiler from $c$. Thus, the length of maximal plays from $c$ consistent with $\pi$ is bounded by some $n \in \mathbb{N}$.\footnote{Here the length of a play is the number of Spoiler configurations that appear in it.}
	  
	  Now, for each $n$, let $\mathcal{W}_S^n$ be the set of configurations $c \in \mathcal{W}_S$ for which there is a winning strategy $\pi$ from $c$
	  such that the length of the longest maximal play consistent with $\pi$ is at most $n$. We prove $\mathcal{W}_S^n \subseteq W$ by induction on $n$. By the above argument, $\mathcal{W}_S = \bigcup_{n \in \mathbb{N}} \mathcal{W}_S^n$, so then we are done.
	
    \begin{enumerate}
      \item In the base case, the longest maximal play consistent with a winning strategy is of the form $c \rightarrow_S c'$ and Duplicator is stuck. Then clearly $c \in \lozenge_S \Box_D W$. Since $\lozenge_S \Box_D W \subseteq W$ we are done.
	  \item For the inductive step, suppose $c \in \mathcal{W}_S^n$. Let $\pi$ be a witnessing strategy where the longest maximal play consistent with $\pi$ has length at most $n$, and suppose $\mathcal{W}_S^{n-1} \subseteq W$.
	  Take any $c'$ such that $\pi(c) \rightarrow_D c'$. Then $\pi$ is winning in $c'$ and every maximal play consistent with it has length at most $n-1$, so $c' \in \mathcal{W}_S^{n-1}$. By the induction hypothesis we get $c' \in W$. Since $c \rightarrow_S \pi(c) \rightarrow_D c'$ we get $c \in \lozenge_S \Box_D W$ and since $\lozenge_S \Box_D W \subseteq W$ we get $c\in W$. \qed
    \end{enumerate}  	
	
\end{proof}

\begin{remark}
	The above proof seems non-constructive: we used K\"onig's lemma in a proof by contradiction, to argue that maximal plays consistent with a Spoiler winning strategy are bounded. This boundedness is based on the assumption that the LTS is image-finite. To avoid this use of K\"onig's lemma, one can assume the LTS to be \emph{finite state}; then the length of maximal plays consistent with a winning strategy for Spoiler are bounded by the number of possible configurations.
\end{remark}

The fixed-point characterisation of the winning region of Spoiler is very close to apartness. Indeed, instantiating the fixed-point characterisation for the winning region of Spoiler with the moves of the strong bisimulation game, we recover the definition of apartness. Thus we obtain:

\begin{theorem}\label{thm:apart-spoiler}
  Two states $x,y$ of an LTS are apart (i.e. $x \apart y$) iff Spoiler has a winning strategy from the configuration $[x,y]$ in the strong bisimulation game.
\end{theorem}
The dual version, observed by Stirling, can be similarly obtained from the fixed-point characterisation of the winning region of Duplicator.

\begin{theorem}[\!\!{\cite[Proposition~1]{Stirling92:TReport:bisimGames}}]\label{thm:stirling-strong-bis}
	States $x,y \in X$ are bisimilar iff $[x,y]$ is winning for Duplicator in the strong bisimulation game.
\end{theorem}

We recover the following from \cref{thm:stirling-strong-bis},~\cref{thm:apart-spoiler} and~\cref{thm:complement-strong-bis}.
\begin{corollary}[\!\!\cite{Stirling99-BisimGames}]
	The bisimulation game is determined, i.e., every configuration is winning for exactly one of the two players.
\end{corollary}
%
%

\section{Branching bisimilarity}

In this section, we study LTSs with silent (internal) actions and consider branching bisimilarity~\cite{BBisim:96} as the notion of behavioural equivalence. We recall the branching bisimulation game of~\cite{YinFHHT14} and connect Spoiler winning positions to the notion of branching apartness~\cite{GeuversJacobs21-apartness}.

Throughout this section, we again fix an LTS $(X,A,\rightarrow)$ and assume that the set of labels contains $A$ contains a distinguished silent action $\tau \in A$. We use $\alpha, \beta$
to range over $A$, and $a,b$ for labels in $A \setminus \{\tau\}$.
We write $x \Longrightarrow x'$ if there is a sequence of $\tau$ steps from $x$ to $x'$.
We use $x \xrightarrow{(\alpha)} x'$ to denote that either (1) $x \xrightarrow{\alpha} x'$ or (2) both $\alpha = \tau$ and $x=x'$.
We assume that $X$ is \emph{finite}; as a consequence, $\Longrightarrow$ is finitely branching, that is, for each $x \in X$, the set $\{x' \mid x \Longrightarrow x'\}$ is finite.

\begin{definition}
  A symmetric relation $R\subseteq X \times X$ is a \emph{branching bisimulation} if for all $(x,y) \in R$:
  \begin{itemize}
	  \item
  if $x \xrightarrow{\alpha} x'$ then $\exists y',y''.\, y \Longrightarrow y' \xrightarrow{(\alpha)} y'' \wedge x R y' \wedge x' R y''$.
 \end{itemize}
  Two states $x,y\in X$ are \emph{branching bisimilar}, denoted $x \bbisim y$, iff there exists a bisimulation $R$ such that $x \mathrel R y$.
\end{definition}
There is, accordingly, a natural notion of apartness~\cite{GeuversJacobs21-apartness}.
\begin{definition}\label{def:branching-apart}
    A \emph{branching apartness relation} $R \subseteq X\times X$ is a symmetric relation satisfying the following rules.
    \[
    \frac{x \xrightarrow{\alpha} x' \quad \forall y',y''. \, y \Longrightarrow y' \xrightarrow{(\alpha)} y'' \text{ implies }  (x \mathrel R y' \vee x' \mathrel R y'')}{x \mathrel R y}
    \]
	As usual we say $x,y$ are \emph{branching apart}, denoted $x \mathrel{\bapart} y$, iff they are related by every branching apartness relation $R$.
\end{definition}

\begin{example}
  Consider the LTS with silent action as given below.

$$
  
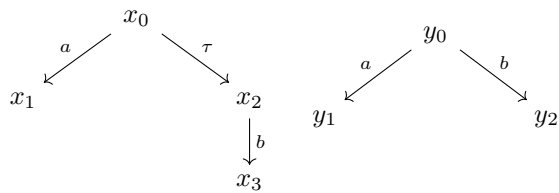
\begin{tikzcd}
    & \arrow[ld,"a"'] x_0 \arrow[rd,"\tau"] & \\
    x_1& & x_2\arrow[d,"b"]\\
    & & x_3
  \end{tikzcd}
  \quad
  \begin{tikzcd}
    & \arrow[ld,"a"'] y_0 \arrow[rd,"b"] & \\
    y_1& & y_2\\
  \end{tikzcd}
$$

  The states $x_0$ and $y_0$ are branching apart, as can be shown with the following proof tree.
  \[
  \infer{x_0 \mathrel{\bapart} y_0}{x_0 \xrightarrow\tau x_2 & \forall y,y'. y_0 \text{ implies } y \xrightarrow{(\tau)} y' \ \text{implies}\ \infer{(x_0 \mathrel{\bapart} y \lor x_2 \mathrel{\bapart} y')}
  {(x_0 \mathrel{\bapart} y_0 \lor
  \infer{x_2 \mathrel{\bapart} y_0)}
  {\infer{y_0 \mathrel{\bapart} x_2}{y_0 \xrightarrow a y_1}}
  }
  }
  \]
\end{example}

The following game for branching bisimilarity comes from~\cite{YinFHHT14}, which is another turn-based two-player game. However, in the branching bisimulation game, Spoiler has two types of Spoiler moves or transition relations from its configurations. Contrary to the (strong) bisimulation game, the following game is no longer alternating, although it is ``almost'': after the first move from Spoiler, single moves of Duplicator are alternated with two consecutive moves from Spoiler. \begin{definition}[Branching bisimulation game]
   The set of configurations is given by $C = C_S \uplus C_D$, where $C_S = X^2 \cup X^5$ are Spoiler configurations, ranged over by tuples denoted by $[x,y]$ and $[x,x',y,y',y'']$ respectively; and $C_D = X \times A \times X \times X$ are Duplicator configurations, ranged over by tuples of the form $\langle x,a,x',y \rangle$. The moves are:
   \begin{itemize}
     \item Spoiler can move as follows:
	 \begin{enumerate}
	 	\item from $[x,y]$ to $\langle x,\alpha,x',y \rangle$ if there is a transition of the form  $x \xrightarrow{\alpha} x'$;
	 	\item from $[x,y]$ to $\langle y,\alpha,y',x \rangle$ if there is a transition of the form  $y \xrightarrow{\alpha} y'$;
		\item from $[x,x',y,y',y'']$ to $[x,y']$ or to $[x',y'']$;
	\end{enumerate}
     \item Duplicator can move from a configuration $\langle x,\alpha,x',y \rangle$ to $[x,x',y,y',y'']$ if there exist transitions of the form $y \Longrightarrow y' \xrightarrow{(\alpha)} y'$.
   \end{itemize}
   We model the Spoiler moves in Items 1 and 2 as a relation $\rightarrow_{S,1} \subseteq X^2 \times C_D$, and the Spoiler moves in Item 3 as a relation $\rightarrow_{S,2} \subseteq X^5 \times X^2$. Similarly, we write $c \rightarrow_D c'$ if there is a move from $c$ to $c'$ by Duplicator.
\end{definition}
   Plays and winning configurations are defined as in the strong bisimulation game.

In the branching bisimilarity game, Duplicator can answer with a sequence of $\tau$-steps followed by an actual $\alpha$-transition (or no transition at all, if $\alpha=\tau$). The key idea is to return the relevant information of this answer to Spoiler: the state just before and just after the $\alpha$-transition. Spoiler can then choose which one to proceed with.

\begin{remark}
	As we assume that $X$ is finite, the set of Duplicator moves is also finite. Note that, contrary to the previous section, it does not suffice to assume image-finiteness of the original LTS to ensure this, since Duplicator can use $\Longrightarrow$.
\end{remark}

\begin{remark}\label{rem:one-step-games}
A more recent game characterisation of branching bisimilarity is proposed by de Frutos-Escrig et al.~\cite{Frutos-EscrigKW17}. That game has the advantage of being \emph{local}: moves are defined from single steps in the transition system,
as opposed to the above game, where Duplicator can respond with the transitive closure. To deal with divergence, the winning condition then includes a non-trivial liveness property. A correspondence between such games and apartness is left for future work; this could perhaps be based on an adapted ``one-step'' apartness rule.
\end{remark}


Just like in the case of strong bisimulation, we provide a fixed-point characterisation of the winning regions of Duplicator and Spoiler. Notice the variety of box and diamond modalities, one for each type of Spoiler moves $\rightarrow_{S,i}$ (for $i\in\{1,2\}$). We focus on Spoiler only.

\begin{proposition}\label{prop:BBisim-WinSet}
	In the branching bisimulation game, the winning region $\mathcal{W}_S$ of Spoiler is the least set $W$ such that
        \[ \lozenge_{S,1} \Box_D \lozenge_{S,2} W \subseteq W \,.\]
\end{proposition}
The proof is analogous to that of \cref{prop:WinRegions}, with the key difference being that Spoiler has the ``extra'' move after every Duplicator move. When proving that $\lozenge_{S,1} \Box_D \lozenge_{S,2} \mathcal{W}_D \subseteq \mathcal{W}_D$, one extends the winning strategies with a first step as before; but this now includes two moves from Spoiler. In the proof that $\mathcal{W}_D \subseteq W$ whenever $ \lozenge_{S,1} \Box_D \lozenge_{S,2} W \subseteq W$, the key is that Duplicator has finitely many possible moves, since the original LTS is finite; this means that the notion of longest maximal play consistent with a Spoiler strategy is once again well-defined.

The above characterisation helps in establishing a correspondence between Spoiler winning configurations and apartness proofs.
\begin{theorem}
	For any $x,y \in X$, we have that $x$ and $y$ are branching apart iff $[x,y]$ is winning for Spoiler in the branching bisimulation game.
\end{theorem}
\begin{proof}
	We prove the implication from left to right by induction on the proof tree of $x \mathrel\bapart y$. Note that symmetry of the apartness relation corresponds to the Spoiler player having the choice which state to play from a pair.
Suppose the rule is applied to conclude $x \mathrel\bapart y$, with premise
	$x \xrightarrow{\alpha} x'$, so that for all $y \Longrightarrow y' \xrightarrow{\alpha} y''$ we have $x \mathrel\bapart y'$ or $x'\mathrel\bapart y''$.
	The induction hypothesis tells us that for each of these transitions, $[x,y']$ or $[x',y'']$ are winning for Spoiler (note that the base case is when there are no such transitions, Duplicator is stuck and Spoiler wins immediately). We have to show that $[x,y]$ is winning for Spoiler.

	Indeed, Spoiler can move to $\langle x, \alpha, x', y \rangle$, and Duplicator has to answer $[x,x',y,y',y'']$
	based on a transition
	$y \Longrightarrow y' \xrightarrow{\alpha} y''$. At this point Spoiler can move to $[x,y']$ or $[x',y'']$, one of which is a winning position.

	For the converse, it suffices to show that the apartness relation $\bapart$ satisfies $\lozenge_{S,1} \Box_D \lozenge_{S,2} \bapart \subseteq \bapart$. By \cref{prop:BBisim-WinSet}, we then get the desired implication. Indeed, suppose that
	$c \rightarrow_{S,1} c'$ and for any Duplicator move $c' \rightarrow_D c''$ there exists a Spoiler move $c'' \rightarrow_{S,2} c''' \in C_S$ such that
		$c''' \in \bapart$. Then by the definition of the moves, we can apply the apartness proof rule (possibly first with an application of symmetry) to obtain $c \in \bapart$. \qed
\end{proof}

\section{Future work}

The study of apartness as a dual to bisimilarity has been studied only recently~\cite{GeuversJacobs21-apartness} (although the notion of apartness for coalgebras is older, as explained in \emph{op. cit.} which cites unpublished work from Jacobs written in 1995). In the current work we have connected apartness to games by relating Spoiler strategies with apartness proofs.

One notable limitation of our approach is that we assumed that the underlying LTS is image-finite, and in the case of branching bisimilarity even finite-state. In fact, in the proof of the characterisation of winning regions as a reachability game, we made use of a proof by contradiction and an appeal to K\"onig's lemma just to achieve a usable notion of size that allows us to carry out induction. Notice that in the general case, even apartness proofs will not be finite anymore. We note that one way around the problem of being finite-state might be to adopt the ``one-step'' games of de Frutos-Escrig et al, see \cref{rem:one-step-games}.

We have only analysed strong and branching bisimilarity for LTSs. One direction for future work is to develop similar results for other forms of bisimulation relations like weak bisimulation and branching bisimulation with explicit divergence. The former can already be handled by our results from Section~\ref{sec:strongbisim} by working on the `saturated' transition relation $\twoheadrightarrow\subseteq X\times A^* \times X$ instead of single step transition relation $\rightarrow$. In particular, $x\stackrel{w}{\twoheadrightarrow}x'$ iff $x'$ is reachable from $x$ under observation $w\in A^*$ with $\tau$-steps interspersed between each observable step in $A$. As long as the state space is finite (the restriction required in the section on branching bisimilarity), Proposition~\ref{prop:WinRegions} remains applicable.

Another direction for future is to try and extend these ideas to a more general coalgebraic framework. Thus, it would be interesting to relate apartness to existing work on coalgebraic games~\cite{FordMSB022,KomoridaKHKH19-codensitygames,DBLP:conf/cmcs/0001MS20}; in particular, the work \cite{DBLP:conf/cmcs/0001MS20,WissmannMS22} that explicitly connect Spoiler strategies to distinguishing formulas. For instance, in \cite{DBLP:conf/cmcs/0001MS20}, the authors give procedures to compute Spoiler strategies for a bisimulation game and construction of a distinguishing formula from a Spoiler strategy (both at the levels of coalgebras). Moreover, a Spoiler strategy is given by a pair of functions (instead of an apartness proof): one modelling the smallest index when two states are separated in the fixed-point computation of bisimilarity; while the other encodes the moves of Spoiler from a given pair of states in the bisimulation game.
We leave the general coalgebraic study of the connection between apartness, distinguishing formulas and games for future work.

\bibliographystyle{plain}

\bibliography{refs}

\begin{thebibliography}{10}

\bibitem{Frutos-EscrigKW17}
David de~Frutos{-}Escrig, Jeroen J.~A. Keiren, and Tim A.~C. Willemse.
\newblock Games for bisimulations and abstraction.
\newblock {\em Log. Methods Comput. Sci.}, 13(4), 2017.

\bibitem{FordMSB022}
Chase Ford, Stefan Milius, Lutz Schr{\"{o}}der, Harsh Beohar, and Barbara
  K{\"{o}}nig.
\newblock Graded monads and behavioural equivalence games.
\newblock In {\em {LICS}}, pages 61:1--61:13. {ACM}, 2022.

\bibitem{Geuvers2022}
Herman Geuvers.
\newblock Apartness and distinguishing formulas in {Hennessy-Milner} logic.
\newblock In {\em A Journey from Process Algebra via Timed Automata to Model
  Learning}, volume 13560 of {\em {LNCS}}, pages 266--282. Springer, 2022.

\bibitem{GeuversJacobs21-apartness}
Herman Geuvers and Bart Jacobs.
\newblock Relating apartness and bisimulation.
\newblock {\em Log. Methods Comput. Sci.}, 17(3), 2021.

\bibitem{hm:hm-logic}
Matthew Hennessy and Robin Milner.
\newblock Algebraic laws for nondeterminism and concurrency.
\newblock {\em Journal of the ACM}, 32:137--161, 1985.

\bibitem{KomoridaKHKH19-codensitygames}
Yuichi Komorida, Shin{-}ya Katsumata, Nick Hu, Bartek Klin, and Ichiro Hasuo.
\newblock Codensity games for bisimilarity.
\newblock In {\em {LICS}}, pages 1--13. {IEEE}, 2019.

\bibitem{DBLP:conf/cmcs/0001MS20}
Barbara K{\"{o}}nig, Christina Mika{-}Michalski, and Lutz Schr{\"{o}}der.
\newblock Explaining non-bisimilarity in a coalgebraic approach: Games and
  distinguishing formulas.
\newblock In {\em {CMCS}}, volume 12094 of {\em {LNCS}}, pages 133--154.
  Springer, 2020.

\bibitem{Park81-bisim}
David Michael~Ritchie Park.
\newblock Concurrency and automata on infinite sequences.
\newblock In {\em Theoretical Computer Science}, volume 104 of {\em {LNCS}},
  pages 167--183. Springer, 1981.

\bibitem{Stirling92:TReport:bisimGames}
Colin Stirling.
\newblock Modal and temporal logics for processes.
\newblock LFCS ECS-LFCS-92-221, The University of Edinburgh, 1992.

\bibitem{Stirling99-BisimGames}
Colin Stirling.
\newblock Bisimulation, modal logic and model checking games.
\newblock {\em Log. J. {IGPL}}, 7(1):103--124, 1999.

\bibitem{TurkenburgBKR23}
Ruben Turkenburg, Harsh Beohar, Clemens Kupke, and Jurriaan Rot.
\newblock Forward and backward steps in a fibration.
\newblock In {\em {CALCO}}, volume 270 of {\em LIPIcs}, pages 6:1--6:18.
  Schloss Dagstuhl - Leibniz-Zentrum f{\"{u}}r Informatik, 2023.

\bibitem{GlabbeekSpectrumII}
Rob~J. van Glabbeek.
\newblock The linear time --- branching time spectrum ii.
\newblock In Eike Best, editor, {\em CONCUR'93}, pages 66--81, Berlin,
  Heidelberg, 1993. Springer Berlin Heidelberg.

\bibitem{BBisim:96}
Rob~J. van Glabbeek and W.~Peter Weijland.
\newblock Branching time and abstraction in bisimulation semantics.
\newblock {\em J. ACM}, 43(3):555–600, May 1996.

\bibitem{WissmannMS22}
Thorsten Wi{\ss}mann, Stefan Milius, and Lutz Schr{\"{o}}der.
\newblock Quasilinear-time computation of generic modal witnesses for
  behavioural inequivalence.
\newblock {\em Log. Methods Comput. Sci.}, 18(4), 2022.

\bibitem{YinFHHT14}
Qiang Yin, Yuxi Fu, Chaodong He, Mingzhang Huang, and Xiuting Tao.
\newblock Branching bisimilarity checking for {PRS}.
\newblock In {\em {ICALP} {(2)}}, volume 8573 of {\em {LNCS}}, pages 363--374.
  Springer, 2014.

\end{thebibliography}

\end{document}